\documentclass[conference]{IEEEtran}
\ifCLASSINFOpdf
\else
\fi
\usepackage{graphicx}
\usepackage{array}
\usepackage{amsmath}
\usepackage{float}
\usepackage{float}
\usepackage{multirow}
\usepackage{array}
\usepackage{cancel}
\usepackage[normalem]{ulem}
\usepackage{tabu}
\usepackage{amsthm}
\usepackage{epsfig}
\usepackage{epstopdf}
\usepackage{epsf}
\usepackage{color}
\usepackage[english]{babel}

\usepackage{url}
\usepackage{color}
\usepackage{xcolor}

\newtheorem{theorem}{Theorem}
\newtheorem{lemma}[theorem]{Lemma}
\newtheorem{definition}{Definition}
\newtheorem{proposition}{Proposition}
\newtheorem{corollary}{Corollary}

\usepackage{graphicx}
\usepackage{array}
\usepackage{amsmath}
\usepackage{float}
\usepackage{float}
\usepackage{multirow}
\usepackage{array}
\usepackage{cancel}
\usepackage[normalem]{ulem}
\usepackage{tabu}
\usepackage{amsthm}
\usepackage{epsfig}
\usepackage{epstopdf}
\usepackage{epsf}
\usepackage{color}
\usepackage[english]{babel}
\usepackage{url}
\usepackage{color}
\usepackage{xcolor}

\begin{document}
%
\title{On Bounds and Closed Form Expressions for Capacities of Discrete Memoryless 
Channels with Invertible Positive Matrices}
%
%
%

\author{\IEEEauthorblockN{Thuan Nguyen}
\IEEEauthorblockA{School of Electrical and\\Computer Engineering\\
Oregon State University\\
Corvallis, OR, 97331\\
Email: nguyeth9@oregonstate.edu}
\and
\IEEEauthorblockN{Thinh Nguyen}
\IEEEauthorblockA{School of Electrical and\\Computer Engineering\\
Oregon State University\\
Corvallis, 97331 \\
Email: thinhq@eecs.oregonstate.edu}
}

%
%

\markboth{IEEE Transactions on Communications}%
{Submitted paper}
%



\maketitle

\begin{abstract}
While capacities of discrete memoryless channels are well studied, it is still not possible to obtain a closed form expression for the capacity of an arbitrary discrete memoryless channel.  This paper describes an elementary technique based on Karush-Kuhn-Tucker (KKT) conditions to obtain (1) a good upper bound of a discrete memoryless channel having an invertible positive channel matrix and (2) a closed form expression for the capacity if the channel matrix satisfies certain conditions related to its singular value and its Gershgorin's disk.
  
\end{abstract}

\begin{IEEEkeywords}
Wireless Communication, Convex Optimization, Channel Capacity, Mutual Information.
\end{IEEEkeywords}

%
\IEEEpeerreviewmaketitle

\section{Introduction}
Discrete memoryless channels (DMC) play a critical role in the early development of  information theory and its applications. DMCs are especially useful for studying many well-known modulation/demodulation schemes (e.g., PSK and QAM ) in which the continuous inputs and outputs of a channel are quantized into discrete symbols.  Thus, there exists a rich literature on the capacities of DMCs \cite{cover2012elements}, \cite{blahut1972computation}, \cite{arimoto1972algorithm}, \cite{muroga1953capacity}, \cite{shannon1956zero}, \cite{robert1990ash},  \cite{nguyen2018closed}.  In particular, capacities of many well-known channels such as (weakly) symmetric channels can be written in elementary formulas \cite{cover2012elements}.  However, it is often not possible to express the capacity of an arbitrary DMC in a closed form expression \cite{cover2012elements}. Recently, several papers have been able to obtain closed form expressions for a small class of DMCs with  small alphabets. For example,  Martin et al. established closed form expression  for a general binary channel \cite{martin2010algebraic}.  Liang showed that the capacity of  channels with two inputs and three outputs can be expressed  as an infinite series \cite{liang2008algebraic}. Paul Cotae et al. found the capacity of  two input and two output channels in term of the eigenvalues of the channel matrices \cite{cotae2010eigenvalue}. On the other hand, the problem of finding the capacity of a discrete memoryless channel can be formulated as a convex optimization problem \cite{grant2008cvx}, \cite{sinha2014convex}. Thus, efficient algorithmic solutions exist.  There is also others algorithms such as Arimoto-Blahut algorithm \cite{blahut1972computation}, \cite{arimoto1972algorithm} which can be accelerated in \cite{dupuis2004blahut}, \cite{matz2004information}, \cite{yu2010squeezing}. In \cite{meister1967capacity}, \cite{jimbo1979iteration}, another iterative method which can yield both upper and lower bounds for the channel capacity.

That said, it is still beneficial to find the channel capacity in closed form expression for a number of reasons.  These include  (1) formulas can often provide a good intuition about the relationship between the capacity and different channel parameters, (2) formulas offer a faster way to determine the capacity than that of algorithms, and (3) formulas are useful for analytical derivations where closed form expression of the capacity is needed in the intermediate steps.  To that end, our paper describes an elementary technique based on the theory of convex optimization, to find closed form expressions for (1) a new upper bound on capacities of discrete memoryless channels with positive invertible channel matrix and (2) the optimality conditions of the channel matrix such that the upper bound is precisely the capacity.  In particular, the optimality conditions establish a relationship between the singular value and the Gershgorin's disk of the channel matrix.

\section{Preliminaries}
\subsection{Convex Optimization and KKT Conditions}
A DMC is characterized by a random variable $X \in \{x_1, x_2, \dots, x_m\}$ for the inputs, a random variable $Y \in \{y_1, y_2, \dots, y_n\}$ for the outputs, and a channel matrix $A \in \mathbf{R}^{m \times n}$. In this paper, we consider DMCs with equal number of inputs and outputs $n$, thus $A \in \mathbf{R}^{n \times n}$. The matrix entry $A_{ij}$ represents the conditional probability that given $x_i$ is transmitted, $y_j$ is received.  Let $p = (p_1, p_2, \dots, p_n)^T$ be the input probability mass vector (pmf) of $X$, where $p_i$ denotes the probability of $x_i$ to be transmitted, then the pmf of $Y$ is $q = (q_1, q_2, \dots, q_n)^T = A^Tp$. The mutual information between $X$ and $Y$ is:
\begin{equation}
I(X;Y) = H(Y) - H(Y|X),
\end{equation}
where
\begin{eqnarray}
 H(Y) &=&-\sum_{j=1}^{n}{q_{j}\log{q_j}} \\
H(Y|X) &=& -\sum_{i=1}^{n}\sum_{j=1}^{n} {p_i A_{ij}} \log {A_{ij}}.
\end{eqnarray}
The mutual information function can be written as:
\begin{equation}
I(X;Y) = -\sum_{j=1}^{n}{(A^Tp)_j\log{(A^Tp)_j}} +\sum_{i=1}^{n}\sum_{j=1}^{n} {p_iA_{ij}} \log {A_{ij}},
\end{equation}
where $(A^Tp)_j$ denotes the $j^{th}$ component of the vector $q = (A^Tp)$.  The capacity $C$  associated with a channel matrix $A$ is the theoretical maximum rate at which information can be transmitted over the channel without the error \cite{shannon1956zero}, \cite{shannon1998mathematical},  \cite{cover1975achievable}.  
It  is obtained using the optimal pmf $p^*$ such that $I(X;Y)$ is maximized. For a given channel matrix $A$, $I(X;Y)$ is a concave function of $p$ \cite{cover2012elements}.  Therefore, maximizing $I(X;Y)$ is equivalent to minimizing $-I(X;Y)$, and finding the capacity can be cast as the following convex problem:

\indent Minimize: 
 \begin{equation}
	\sum_{j=1}^{n}{(A^Tp)_{j}\log{(A^Tp)_j}} -\sum_{i=1}^{n}\sum_{j=1}^{n} {p_i A_{ij}} \log {A_{ij}} \nonumber. \\
\end{equation}
\indent Subject to: 
 $$\begin{cases}
& p \succeq \mathbf{0}\\
& \mathbf{1}^Tp = 1. 
\end{cases}$$

The optimal $p^*$ can be found efficiently using various algorithms such as gradient methods \cite{boyd2004convex}, but in a few cases, $p^*$ can be found directly using the Karush-Kuhn-Tucker (KKT) conditions \cite{boyd2004convex}.  To explain the KKT conditions,  we first state the canonical convex optimization problem below:

Problem \textbf{P1}:
\indent Minimize: $f(x)$ \\
\
\indent Subject to: 
 $$\begin{cases}
& g_i(x) \le  0, i = 1, 2, \dots n,\\
&h_j(x) = 0, j = 1, 2, \dots, m, \\
\end{cases}$$

where $f(x)$, $g_i(x)$ are convex functions and $h_j(x)$ is a linear function.

Define the Lagrangian function as:
\begin{equation}
\label{sec:lagrangian}
L(x,\lambda, \nu) = f(x) + \sum_{i=1}^n{\lambda_i g_i(x)} + \sum_{j=1}^m{\nu_j h_j(x)},
\end{equation}

then the KKT conditions \cite{boyd2004convex} states that, the optimal point $x^*$ must satisfy:

\begin{equation}
\label{eq:kkt1}
\begin{cases}
g_i(x^*) \le 0, \\ 
h_j(x^*) = 0, \\
\frac{d{L(x, \lambda, \nu)}}{dx}|_{x = x^*, \lambda = \lambda^*, \nu = \nu^*} = 0, \\
\lambda_i^*g_i(x^*) = 0, \\
\lambda_i^* \ge 0.
\end{cases}
\end{equation}
for $i = 1, 2, \dots, n$, $j = 1, 2, \dots, m$.

\subsection{Elementary Linear Algebra Results}

\begin{definition}
	\label{def:K}
	Let $A \in \mathbf{R}^{n \times n}$ be an invertible channel matrix and $H(A_i) = -\sum_{k=1}^{n} {A}_{ik}\log{A}_{ik}$ be the entropy of $i^{th}$ row, define
	$${K}_{j} = -\sum_{i=1}^{n} {{A}_{ji}^{-1}} \sum_{k=1}^{n} {A}_{ik}\log{A}_{ik} = \sum_{i=1}^{n} {{A}_{ji}^{-1}} H(A_i),$$
	where ${A}_{ji}^{-1}$ denotes the entry $(j,i)$ of the inverse matrix $A^{-1}$.  $K_{\max} = \max_j{K_j}$ and $K_{\min} = \min_j{K_j}$ are called the maximum and minimum  inverse row entropies of $A$, respectively. 
\end{definition}

\begin{definition}
	\label{def:gershgorin} 
	Let $A \in \mathbf{R}^{n \times n}$ be a square matrix.  The Gershgorin radius of $i^{th}$ row of $A$ \cite{weisstein2003gershgorin} is defined as:
	\begin{equation}
	R_i(A) = \sum_{j \neq i}^n{|A_{ij}|}.
	\end{equation}
	The Gershgorin ratio of $i^{th}$ row of $A$ is defined as:
	\begin{equation}
	c_i(A) = \frac{A_{ii}}{R_i(A)},
	\end{equation}
	and the minimum Gershgorin ratio of $A$ is defined as:
	\begin{equation}
	\label{eq: c min A}
	c_{\min}(A) = \min_i{\frac{A_{ii}}{R_i(A)}}.
	\end{equation}
\end{definition}
We note that since the channel matrix is a stochastic matrix, therefore 
\begin{equation}
c_{\min}(A) = \min_i{\frac{A_{ii}}{R_i(A)}}= \min_i{\frac{A_{ii}}{1-A_{ii}}}. \label{eq: find c min}
\end{equation}

\begin{definition}
	\label{def:non-negative_diag_dom}
	Let $A \in \mathbf{R}^{n \times n}$ be a square matrix. 
	
(a) $A$ is called a positive matrix if $A_{ij}>0$ for $\forall$ $i,j$. 

(b) $A$ is called a strictly diagonally dominant positive matrix \cite{fiedler1967diagonally} if $A$ is a positive matrix and
\begin{equation}
\label{eq: condition for diagonal}
A_{ii} > \sum_{j \neq i}{A_{ij}}, \forall i, j.
\end{equation}
\end{definition}



\begin{lemma}
	\label{lemma:dominant}
	Let $A \in \mathbf{R}^{n \times n}$ be a strictly diagonally dominant positive channel matrix 
	then (a) it is invertible; (b) the eigenvalues of $A^{-1}$ are  $\frac{1}{\lambda_i}$ $\forall$ $i$ where $\lambda_i$ are eigenvalues of $A$, (c) $A^{-1}_{ii} > 0$ and the largest absolute element in the $i^{th}$ column of $A^{-1}$ is $A^{-1}_{ii}$, i.e.,  $A^{-1}_{ii} \ge |A^{-1}_{ji}|$ for $\forall$ $j$.
\end{lemma}

\begin{proof}
	The proof is shown in Appendix \ref{sec: appendix for lemma 1}. 
\end{proof}

\begin{lemma}
	\label{lemma:c_i}
	Let $A \in \mathbf{R}^{n \times n}$ be a strictly diagonally dominant positive matrix, then:
	\begin{equation}
	\label{eq: thuan1419}
	c_i(A^{-T}) \ge \frac{c_{\min}(A)-1}{(n-1)}, \forall i.
	\end{equation}
	Moreover, for any rows $k$ and $l$,
	\begin{eqnarray}
\label{eq: thuan1418}
|A^{-1}_{ki}|+|A^{-1}_{li}| &\leq &  A_{ii}^{-1}\dfrac{c_{\min}(A)}{c_{\min}(A)-1}, \forall i.
\end{eqnarray}

\end{lemma}
\begin{proof}
	The proof is shown in Appendix \ref{sec: appendix for lemma 2}.
\end{proof}

\begin{lemma}
	\label{lemma: singular}
	Let $A \in \mathbf{R}^{n \times n}$ be a strictly diagonally dominant positive matrix, then:
	\begin{equation}
	\label{eq: singular min}
	\max_{i,j}{A^{-1}_{ij}} \leq \dfrac{1}{{\sigma_{\min}(A)}},
	\end{equation}
where $\max_{i,j}{A^{-1}_{ij}}$ is the largest entry in $A^{-1}$ and ${\sigma_{\min}(A)}$ is the minimum singular value of $A$.
\end{lemma}
\begin{proof}
	The proof is shown in Appendix \ref{sec: appendix for lemma 3}.
\end{proof}

\begin{lemma}
	\label{prop:sum_rowAinv=1}
	Let $A \in \mathbf{R}^{n \times n}$ be an invertible channel matrix, then
	$$A^{-1}\textbf{1} = \textbf{1},$$
	i.e., the sum of any row of $A^{-1}$ equals to 1.  
	Furthermore, for any probability mass vector $x$, sum of the vector $y = {A^{-T}}x$ equal to 1.
\end{lemma}
\begin{proof}
	The proof is shown in Appendix \ref{appendix for pmf}.
\end{proof}
\section{Main Results}

\label{sec:upperbound}

Our first main result is an upper bound on the capacity of discrete memoryless channels having invertible positive channel matrices. 

\begin{proposition} [Main Result 1]
\label{prop:upperbound}
Let $A \in \mathbf{R}^{n \times n}$ be an invertible positive channel matrix and 
\begin{equation}
\label{eq: closed form for q optimal}
q^*_j = \frac{2^{-K_j}}{\sum_{i=1}^n{2^{-K_i}}},
\end{equation} 
\begin{equation}
p^* = A^{-T}q^*,
\end{equation} then the capacity $C$ associated with the channel matrix $A$ is upper bounded by:
\begin{equation}
\label{eq: upper bound capacity}
C \le -\sum_{j=1}^n{q^*_j\log{q^*_j}} + \sum_{i=1}^{n}\sum_{j=1}^{n} {p_i ^*A_{ij}} \log {A_{ij}}.
\end{equation}
\end{proposition}

\begin{proof}
Let $q$ be the pmf of the output $Y$, then $q = A^{-T}p$. Thus,
\begin{eqnarray}
\label{eq:mutual}
I(X;Y) &=& H(Y)-H(Y|X) \\
&=& -\sum_{j=1}^n{q_j\log{q_j}} + \sum_i^n{(A^{-T}q)_{i}\sum_{k}^n{A_{ik}\log{A_{ik}}}} \nonumber.
\end{eqnarray}

We construct the Lagrangian in  (\ref{sec:lagrangian}) using $-I(X;Y)$ as the objective function and optimization variable $q_j$:
\begin{equation}
L(q_j,\lambda_j,\nu_j)= - I(X;Y) - \sum_{j=1}^{n}{{q_j}{\lambda_j}} + \nu (\sum_{j=1}^{n}{q_j}-1), 
\end{equation}

where the constraints $g(x)$ and $h(x)$ in problem $\textbf{P1}$ are translated into $-{q_j}\leq 0$ and $\sum_{j=1}^{n}{q_j}=1$, respectively.

Using the KKT conditions in  (\ref{eq:kkt1}), the optimal points $q_{j}^{*}$, $\lambda_{j}^{*}$, $\nu^{*}$ for all $j$, must satisfy:  
\begin{eqnarray}
q_{j}^{*} \geq 0, \label{eq:cond1} \\
\sum_{j=1}^{n}q_{j}^{*} = 1, \label{eq:cond2} \\
\nu^{*} - \lambda_{j}^{*} - \dfrac{dI(X;Y)}{dq_{j}^{*}}= 0, \label{eq:cond3}\\
\lambda_{j}^{*} \geq 0, \label{eq:cond4}\\
\lambda_{j}^{*} q_{j}^{*}=0. \label{eq:cond5}
\end{eqnarray}

Since $0 \leq p_i \leq 1 $ and $\sum_{i=1}^{n}{p_i} = 1$, there exists at least one $p_i > 0$ .  Since $A_{ij} > 0$ $\forall i,j$, we have:
\begin{equation}
\label{eq:q_j}
q^*_j = \sum_{i=1}^{n}p^*_i{A}_{ij} > 0, \forall j.
\end{equation}
Based on (\ref{eq:cond5}) and (\ref{eq:q_j}), we must have $\lambda_j^{*} = 0, \forall j$. Therefore, all five KKT conditions (\ref{eq:cond1}-\ref{eq:cond5}) are reduced to the following two conditions:
\begin{eqnarray}
\sum_{j=1}^{n}q_{j}^{*} = 1, \label{eq:s_cond1}\\
\nu^{*} - \dfrac{dI(X;Y)}{dq_{j}^{*}}= 0. \label{eq:s_cond2}
\end{eqnarray}

Next,
\begin{eqnarray}
\dfrac{dI(X;Y)}{dq_{j}} &=& \sum_{i=1}^{n} {{A}_{ji}^{-1}} \sum_{k=1}^{n} {A}_{ik}\log{A}_{ik}- (1+\log{q_j}) \nonumber\\ 
													&=&-K_j - (1 + \log{q_j}).  \label{eq:K}
\end{eqnarray}

Using (\ref{eq:s_cond2}) and (\ref{eq:K}), we have:
\begin{equation}
\label{eq:q_asterisk}
q_j^* = 2^{{-K}_j-\nu^* - 1}.
\end{equation}

Plugging (\ref{eq:q_asterisk}) to (\ref{eq:s_cond1}), we have:
$$\sum_{j=1}^{n} 2^{{-K}_j-\nu^*-1} = 1,$$
$$\nu^*= \log{\sum_{j=1}^{n}{2^{{-K}_j-1}}}.$$
From (\ref{eq:q_asterisk}), 
\begin{equation}
\label{eq:q}
q_j^* = 2^{{-K}_j-\nu^* -1} = \dfrac{2^{-K_j}}{2^{\nu^*+1}} = \dfrac{2^{-K_j}}{\sum_{j=1}^{n} 2^{{-K}_j}}, \forall j.
\end{equation}
If $q^*$ is such that $p^* = A^{-T}q^* \succeq 0$ and $(A^{-T}q^*)^T \mathbf{1}  = \sum_{i}^{n}p_i^*= 1 $, then $p^*$ is a valid p.m.f and Proposition \ref{prop:upperbound} will hold with equality by the KKT conditions. However these two constraints might not hold in general. On the other hand, maximizing $I(X;Y)$ in terms of $q$ and ignoring these constraints is equivalent to enlarging the feasible region, will necessarily yield a value that is at least equal to the capacity $C$.  Thus, by plugging $q^*$ into  (\ref{eq:mutual}), we obtain the proof for the upper bound. 
\end{proof}

Next, we present some sufficient conditions on the channel matrix $A$ such that its capacity can be written in closed form expression. We note that the channel capacity closed form expression is also discovered in \cite{muroga1953capacity} and \cite{robert1990ash} using the input distribution variables. However in both \cite{muroga1953capacity} and \cite{robert1990ash}, the sufficient conditions for closed form expression are not fully characterized.
%
%

\begin{proposition}[Main Result 2]
\label{prop:c_min_K_max_K_min}
Let $A \in \mathbf{R}^{n \times n}$ be a strictly diagonally dominant positive matrix, if $\forall i$,
\begin{equation}
\label{eq: thinh1}
c_{i}(A^{-T}) \ge (n-1) 2^{K_{\max}-K_{\min}},
\end{equation}
then the capacity of channel matrix $A$ admits a closed form expression which is exactly the upper bound in Proposition \ref{prop:upperbound}.  
\end{proposition}
\begin{proof}
Based on the discussion of the KKT conditions, it is sufficient to show that if $p^* = A^{-T}q^* \succeq 0$ and $\sum_{i}^{n}p_i^* = (A^{-T}q^*)^T \mathbf{1}  = 1 $ then $C$ has a closed form expression.  The condition $(A^{-T}q^*)^T \mathbf{1}  = 1$ is always true as shown in Lemma \ref{prop:sum_rowAinv=1} in the Appendix \ref{appendix for pmf}. Thus, we only need to show that if $c_{i}(A^{-T}) \ge 2^{K_{\max}-K_{\min}}$,  then $p^*=A^{-T}q^* \succeq 0.$

Let $q^*_{\min} = \min_j{q^*_j}$ and $q^*_{\max} = \max_j{q^*_j}$, we have:
\begin{eqnarray}
p^*_i &=& \sum_j{q^*_jA^{-1}_{ji}}\nonumber\\
&=& q^*_iA_{ii}^{-1} + \sum_{j\neq i}{q^*_jA^{-1}_{ji}}\nonumber \\
&\ge& q^*_{\min}A^{-1}_{ii} - (\sum_{j\neq i} q^*_j)(\sum_{j \neq i}{|A^{-1}_{ji}|})\label{eq: 1000}\\
&\ge& q^*_{\min}A^{-1}_{ii} - (n-1)q^*_{\max}(\sum_{j \neq i}{|A^{-1}_{ji}|})\label{eq: 1001} ,
\end{eqnarray} 
with (\ref{eq: 1000}) due to $A_{ii}^{-1} >0$ which follows by Lemma \ref{lemma:dominant}-$c$, (\ref{eq: 1001}) is due to $q^*_{\max} \geq q_j^*$ $\forall$ $j$. 
Now if we want $p^*_i \ge 0,$ $\forall$ $i$, from  (\ref{eq: 1001}), it is sufficient to require that, $\forall i$,
\begin{eqnarray}
c_i(A^{-T})= \frac{A^{-1}_{ii}}{\sum_{j \neq i}{|A^{-1}_{ji}|}} &\geq& \frac{(n-1)q^*_{\max}}{q^*_{\min}} \nonumber\\
&=&(n-1)\dfrac{\dfrac{2^{-K_{\min}}}{\sum_{j=1}^{n} 2^{{-K}_j}}}{\dfrac{2^{-K_{\max}}}{\sum_{j=1}^{n} 2^{{-K}_j}}} \label{eq: simple 2}\\
&=&(n-1)2^{K_{\max}-K_{\min}} \nonumber,
\end{eqnarray}
with (\ref{eq: simple 2}) due to (\ref{eq:q}) and $q^*_{\max}$,  $q^*_{\min}$ are corresponding to $K_{\min}$, $K_{\max}$, respectively.
Thus, Proposition \ref{prop:c_min_K_max_K_min} is proven. 

\end{proof}
We are now ready to state and prove the third main result that characterizes the sufficient conditions on a channel 
matrix so that the upper bound in Proposition \ref{prop:upperbound} is precisely the capacity.

\begin{proposition}[Main Result 3]
\label{prop:condition}
Let $A \in \mathbf{R}^{n \times n}$ be a  strictly diagonally dominant positive channel matrix and $H_{\max}(A)$ be the maximum row entropy of $A$. The capacity $C$ is the upper bound in Proposition \ref{prop:upperbound} i.e., hold with equality if 
\begin{large} 
\begin{equation}
\sqrt[\leftroot{-2}\uproot{2} V ]{\dfrac{c_{\min}(A)-1}{(n-1)^2}} \geq 2^{\frac{n H_{\max}(A)}{{\sigma_{\min}(A)}}},
\label{eq:eigen1}
\end{equation}
\end{large}
where ${\sigma_{\min}(A)}$ is the minimum singular value of channel matrix $A$, and
\begin{equation}
V=\dfrac{c_{\min}(A)}{c_{\min}(A)-1}.
\end{equation}
\end{proposition}

\begin{proof}
From (\ref{eq: thuan1419}) in Lemma \ref{lemma:c_i} and Proposition \ref{prop:c_min_K_max_K_min}, if we can show that
\begin{equation}
\label{eq: ha2412}
\frac{c_{\min}(A)-1}{(n-1)} \geq (n-1)2^{K_{\max}-K_{\min}},
\end{equation}
then Proposition \ref{prop:condition} is proven. 
 Suppose that $K_{\max}$ and $K_{\min}$ are obtained  at rows $j=L$ and $j=S$, respectively. We note that from (\ref{eq:q}), $q_{\max}=\max_{j}q_j$ and $q_{\min}=\min_{j}q_j$ correspond to $K_{\min}$ and $K_{\max}$, respectively. Thus, from the Definition \textbf{1}, we have:
\begin{small}
\begin{eqnarray}
K_{\max} \!-\! K_{\min} \!&=\!& {\sum_{i=1}^n{A^{-1}_{Li}H(A_i)}} -{\sum_{i=1}^n{A^{-1}_{Si}H(A_i)}} \nonumber \\
&\leq& |{\sum_{i=1}^n{A^{-1}_{Li}H(A_i)}}| + |{\sum_{i=1}^n{A^{-1}_{Si}H(A_i)}}|\label{eq: thuan1412}\\
&\le& |{\sum_{i=1}^n{A^{-1}_{Li}| |H(A_i)}}| + |{\sum_{i=1}^n{A^{-1}_{Si}| |H(A_i)}}| \label{eq: thuan1413}\\
&\leq & H_{\max}(A)  \sum_{i=1}^n (|A^{-1}_{Li}|+|A^{-1}_{Si}| ) \label{eq: thuan1414}\\
&\leq & H_{\max}(A) \sum_{i=1}^n A_{ii}^{-1}\dfrac{c_{\min}(A)}{c_{\min}(A)-1}\label{eq: thuan1415}\\
&\le & nH_{\max}(A)  (\max_{i,j}{A^{-1}_{ij}}) \dfrac{c_{\min}(A)}{c_{\min}(A)-1}  \label{eq: 105} \\
& \le & \frac{nH_{\max}(A)V}{{\sigma_{\min}(A)}} \label{eq: thinh},
\label{eq:lambda}
\end{eqnarray}
\end{small}
where  (\ref{eq: thuan1412}) due to the property of absolute value function, (\ref{eq: thuan1413})  due to Schwarz inequality, (\ref{eq: thuan1414}) due to $H_{\max}(A)$ is the maximum row entropy of $A$,  (\ref{eq: thuan1415}) due to   (\ref{eq: thuan1418}),  (\ref{eq: 105}) due to $\max_{i,j}{A^{-1}_{ij}}$ is the largest entry in $A^{-1}$ and  (\ref{eq:lambda}) is due to Lemma \ref{lemma: singular}. Thus,
\begin{equation}
\label{eq: ha2413}
(n-1)2^{\frac{nH_{\max}(A)V}{{\sigma_{\min}(A)}}} \geq (n-1)2^{K_{\max}-K_{\min}}.
\end{equation}

From (\ref{eq: ha2412}) and (\ref{eq: ha2413}), if 
 \begin{equation}
 \label{eq: nganha2414}
\frac{c_{\min}(A)-1}{(n-1)} \geq (n-1)2^{\frac{nH_{\max}(A)V}{{\sigma_{\min}(A)}}},
\end{equation} 
then the capacity $C$ is the upper bound in Proposition \ref{prop:upperbound}.   (\ref{eq: nganha2414}) is equivalent to (\ref{eq:eigen1}). Thus Proposition \ref{prop:condition} is proven. 
\end{proof}

An easy to use version of Proposition \ref{prop:condition} is stated in Corollary \ref{corollary 1}. 

\begin{corollary}
\label{corollary 1}
The capacity $C$ is the upper bound in Proposition \ref{prop:upperbound} if  
\begin{large} 
\begin{equation}
\dfrac{c_{\min}(A)-1}{(n-1)^2} \geq 2^{\frac{2 n \log n}{{\sigma_{\min}(A)}}}.
\end{equation}
\end{large}
\end{corollary}
\begin{proof}
Similar to Proposition \ref{prop:condition},
\begin{small}
\begin{eqnarray}
K_{\max} \!-\! K_{\min} &=& {\sum_{i=1}^n{A^{-1}_{Li}H(A_i)}} -{\sum_{i=1}^n{A^{-1}_{Si}H(A_i)}} \nonumber \\
&\leq& |{\sum_{i=1}^n{A^{-1}_{Li}H(A_i)}}| + |{\sum_{i=1}^n{A^{-1}_{Si}H(A_i)}}| \label{eq: vy1413}\\
&\le& |{\sum_{i=1}^n{A^{-1}_{Li}| |H(A_i)}}| + |{\sum_{i=1}^n{A^{-1}_{Si}| |H(A_i)}}|  \label{eq: vy1414}\\
&\leq & H_{\max}(A)  \sum_{i=1}^n (|A^{-1}_{Li}|+|A^{-1}_{Si}| )  \label{eq: vy1415}\\
&\le & H_{\max}(A)  n (2 \max_{i,j}{A^{-1}_{ij}})   \label{eq: thuanha1}\\
& \le & \frac{2 n \log n}{{\sigma_{\min}(A)}} \label{eq: corollary 1},
\end{eqnarray}
\end{small}
with (\ref{eq: vy1413}), (\ref{eq: vy1414}), (\ref{eq: vy1415})   are similar to (\ref{eq: thuan1412}), (\ref{eq: thuan1413}), (\ref{eq: thuan1414}), respectively. 
(\ref{eq: thuanha1}) is due to $\max_{i,j}{A^{-1}_{ij}}$ is the largest entry in $A^{-1}$,  (\ref{eq: corollary 1}) due to $ H_{\max}(A) \leq \log n$ and Lemma \ref{lemma: singular}.  
Thus, by changing $\frac{nH_{\max}(A)V}{{\sigma_{\min}(A)}}$ in  (\ref{eq: nganha2414}) by $\frac{2 n \log n}{{\sigma_{\min}(A)}}$, the Corollary \ref{corollary 1} is proven.
\end{proof}



A direct result of Proposition \ref{prop:condition} without using singular value is shown in Corollary \ref{corollary 2}. 

\begin{corollary}
\label{corollary 2}
The capacity $C$ is the upper bound in Proposition \ref{prop:upperbound} if  
\begin{large}
\begin{equation}
\sqrt[\leftroot{-2}\uproot{2} V ]{\dfrac{c_{\min}(A)-1}{(n-1)^2}} \geq 2^{\frac{n  H_{\max}^*(A)}{{\sigma^*}}},
\label{eq:eigen2}
\end{equation}
\end{large}

where,
\begin{equation}
V=\dfrac{c_{\min}(A)}{c_{\min}(A)-1},
\end{equation}
\begin{equation}
\sigma^*=\dfrac{c_{\min}(A) -n/2}{c_{\min}(A)+1},
\end{equation}
\begin{small}
\begin{equation}
H_{\max}^*(A)=\log(c_{\min}(A)+1)+\dfrac{\log(n-1)-c_{\min}(A)\log c_{\min}(A)}{c_{\min}(A)+1}.
\end{equation}
\end{small}
\end{corollary}
\begin{proof}
We will construct the lower bound for $\sigma_{\min}(A)$ and the upper bound for $H_{\max}(A)$. From Lemma \ref{lemma: 4} in Appendix \ref{appendix for upper bound singular value}
\begin{equation}
\sigma_{\min}(A) \geq \dfrac{c_{\min}(A) -n/2}{c_{\min}(A)+1}=\sigma^*,
\end{equation}

and 

\begin{small}
\begin{eqnarray}
H_{\max}(\!A\!)\! &\! \leq \!& \log(c_{\min}(A)\!+\!1)\!+\!\dfrac{\log(n-1)\!-\!c_{\min}(A)\log c_{\min}(A)}{c_{\min}(A)+1} \nonumber\\
&\!=\!&H_{\max}^*(A).
\end{eqnarray}
\end{small}

Therefore

\begin{equation}
\frac{nH_{\max}(A)V}{{\sigma_{\min}(A)}} \leq \frac{n  H_{\max}^*(A)}{{\sigma^*}}.
\end{equation}

Thus, by changing $ \frac{nH_{\max}(A)}{{\sigma_{\min}(A)}}  $ in (\ref{eq:eigen1}) by $\frac{n  H_{\max}^*(A)}{{\sigma^*}}$, the Corollary \ref{corollary 2} is proven. 

We note that, when $c_{\min}(A)$ is relatively larger than the size of matrix $n$, the lower bound of $\sigma_{\min}(A)$ goes to 1. We also note that  (\ref{eq:eigen2}) can be checked efficiency without requiring both $H_{\max}(A)$ and $\sigma_{\min}(A)$ at the expense of a looser upper bound as compare to   (\ref{eq:eigen1}).
\end{proof}

\section{Examples and Numerical Results}
\subsection{Example 1: Reliable Channels}
We illustrate the optimality conditions in Proposition \ref{prop:condition} using a reliable channel having the channel matrix:
\begin{small}
\[
A=
  \begin{bmatrix}
0.95    &     0.01   &      0.04\\
         0.03     &     0.95     &    0.02\\        
         0.02     &    0.02    &     0.96
  \end{bmatrix}.
\]

\end{small}
Here, $n=3$, ${\sigma_{\min}(A)}=0.92424$, $\sigma^*=0.875$ and $H_{\max}(A)=0.33494$, $H_{\max}^*(A)=0.3364$. From Definition \ref{def:gershgorin}, $c_{\min}(A)=19$.  The closed form channel capacity can be readily computed by Proposition \ref{prop:upperbound} since the channel matrix satisfies both conditions in Proposition \ref{prop:condition} and Corollary \ref{corollary 2}. The optimal input and output probability mass vectors are: 
\begin{small}
\[
q^T=
  \begin{bmatrix}
    0.33087   &   0.32806   &   0.34107
  \end{bmatrix},
\]
\[
p^T=
  \begin{bmatrix}
    0.33067   &    0.33480   &   0.33453
  \end{bmatrix}, 
\]
\end{small}
respectively and the capacity is 1.2715. 

In general, for a good channel with $n$ inputs and $n$ outputs whose symbol error probabilities are small, then it is likely that the channel matrix will satisfy the optimality conditions in Proposition \ref{prop:upperbound}.  This is because the diagonal entries $A_{ii}$ (probability of receiving correct the $i^{th}$ symbol) tend to be larger than the sum of other entries in its row (probability of errors), satisfying the property of diagonally dominant matrix.

\subsection{Example 2: Cooperative Relay-MISO Channels}
In this example, we investigate the channel capacity for a class of channels named Relay-MISO (Relay - Multiple Input Single Output). Relay-MISO channel \cite{thuan_thinh_2018} can be constructed by the combination of a relay channel  \cite{cover1979capacity} \cite{rankov2006achievable} and a Multiple Input Single Output channel, as illustrated in Fig. \ref{fig: relay}. 

In a Relay-MISO channel, $n$ senders want to transmit data to a same receiver via $n$ relay base station nodes. The uplink of these senders using wireless links that are prone to transmission errors. Each sender can transmit bit ``0" or ``1" with the probability of bit flipping is $\alpha$,  $0 \leq \alpha \leq 1$. For a simplicity, suppose that $n$ relay channels have the same error probability $\alpha$. Next, all of the relay base station nodes will relay the signal by a reliable channel such as optical fiber cable to a same receiver. The receiver adds all the relay signals (symbols) to produce a single output symbol.

It can be shown that the channel matrix of this Relay-MISO channel \cite{thuan_thinh_2018}  is an invertible matrix of size $(n+1)\times(n+1)$ whose $A_{ij}$ can be computed as:
\begin{small}
\begin{equation*}
\label{eq: 6}
{ {A}_{ij}\!=\!\sum_{s=\max(i\!-\!j,0)}^{s=\min(n\!+\!1\!-\!j,i\!-\!1)}{j\!-\!i\!+\!s\choose n\!+\!1\!-\!i} {s\choose i\!-\!1} \alpha^{j-i+2s} (1\!-\!\alpha)^{n-(j-i+2s)}  }.
\end{equation*}
\end{small}
\begin{figure}
  \centering
  \includegraphics[width=3 in, height = 3 in]{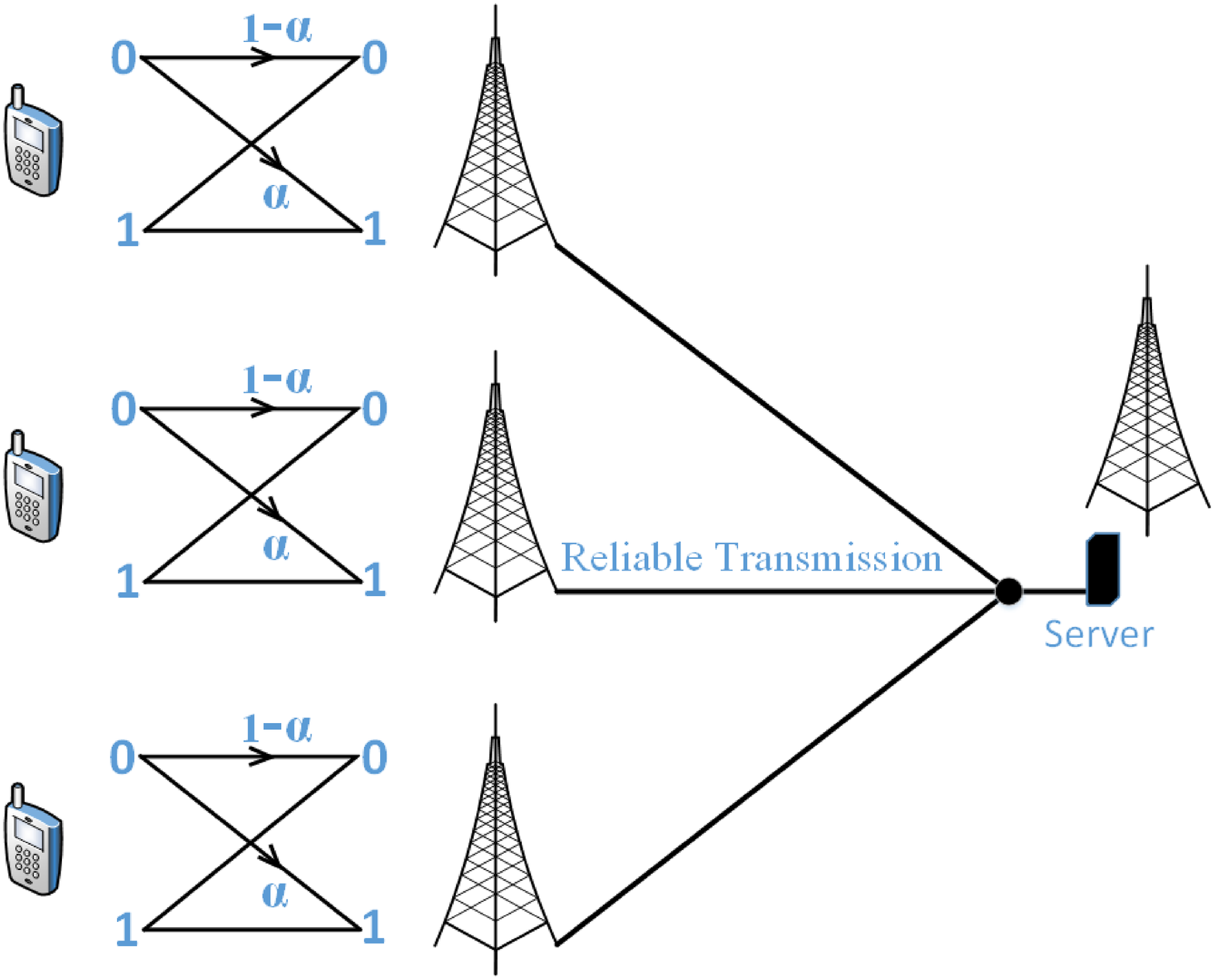}\\
  \caption{Relay-MISO channel}\label{fig: relay}
 \end{figure}
We note that this Relay-MISO channel matrix is invertible and the inverse matrix has the  closed form expression which is characterized in \cite{thuan_thinh_2018}. 
For example, the channel matrix of a Relay-MISO channel with $n=3$ is given as follows:

\begin{scriptsize}
$$\begin{array}{cc}
 \begin{bmatrix}
(1\!-\!\alpha)^3  \!&\! 3 (1\!-\!\alpha)^2 \alpha \!&\! 3(1\!-\!\alpha)\alpha^2 \!&\! \alpha^3\\
\alpha(1\!-\!\alpha)^2 & 2\alpha^2 (1\!-\!\alpha)  +(1\!-\!\alpha)^3 &  2(1\!-\!\alpha)^2 \alpha  \!+\! \alpha^3 & (1\!-\!\alpha)\alpha^2\\
(1\!-\!\alpha)\alpha^2  & 2(1\!-\!\alpha)^2 \alpha  \!+\!\alpha^3 & 2\alpha^2 (1\!-\!\alpha)  \!+\! (1\!-\!\alpha)^3 & \alpha(1\!-\!\alpha)^2 \\
\alpha^3 & 3(1\!-\!\alpha)\alpha^2 & 3 (1\!-\!\alpha)^2 \alpha & (1\!-\!\alpha)^3
\end{bmatrix},
\end{array} $$
\end{scriptsize}

where $0 \leq \alpha \leq 1$. We note that this channel matrix is strictly diagonally dominant matrix when $\alpha$ is close to 0 or $\alpha$ is close to 1. 
In addition, for $\alpha$ values that are close to 0 or 1, it can be shown that channel matrix $A$ satisfies the conditions in Proposition \ref{prop:condition}.  Thus, the channel capacity admits a closed form expression in Proposition \ref{prop:upperbound}.  For other values of $\alpha$, e.g. closer to 0.5, the optimality conditions in Proposition \ref{prop:condition} no longer holds.  In this case, Proposition \ref{prop:upperbound} can still be used as a good upper bound on the capacity.

We show that our upper bound is tighter than existing upper bounds. In particular, Fig. \ref{fig: compare} shows the actual capacity and the known upper bounds as functions of parameter $\alpha$ for Relay-MISO channels having $n=3$.  The green curve depicts the actual capacity computed using convex optimization algorithm.  The red curve is constructed using our closed form expression in Proposition \ref{prop:upperbound}, and the blue dotted curve is  the constructed using the well-known upper bound result of channel capacity in \cite{chiang2004geometric}, \cite{boyd2007tutorial}.  Specifically, this upper bound is:
\begin{equation}
C\leq \log(\sum_{j=1}^{n}\max_{i}A_{ij}).
\end{equation}
%
%

Finally, the red dotted curve shows another well-known upper bound by Arimoto \cite{arimoto1972algorithm} which is:
\begin{equation}
C\leq \log(n) + \max_{j}[\sum_{i=1}^{n} A_{ji} \log(\dfrac{A_{ji}}{\sum_{k=1}^{n}A_{ki}})].
\end{equation}

We note that the second term is negative.

Fig. \ref{fig: compare} shows that our closed form upper bound is precisely the capacity (the red and green graphs are overlapped) when $\alpha$ values are close to 0 or 1 as predicted by the optimality conditions in Proposition \ref{prop:condition}.  On the other hand, when $\alpha$ values are closer to 0.5, our optimality conditions no longer hold.  In this case, we can only determine the upper bound.  However, it is interesting to note that our upper bound in this case is tighter than both the Boy-Chiang \cite{chiang2004geometric} and Arimoto \cite{arimoto1972algorithm} upper bounds.
\begin{figure}
  \centering
  \includegraphics[width=3 in, height = 3 in]{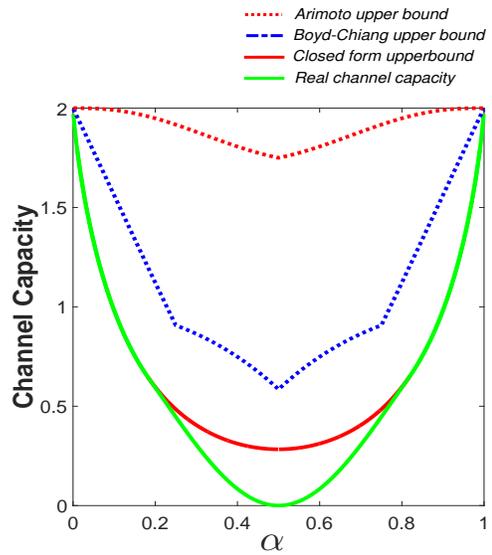}\\
  \caption{Channel capacity and various upper bounds as functions of $\alpha$}\label{fig: compare}
 \end{figure}

\subsection{Example 3: Symmetric and Weakly Symmetric Channels}
\label{sec: arimoto}

Our results confirm the capacity of the well known symmetric and weakly symmetric channel matrices.
In particular, when the channel matrix is symmetric and positive definite, all our results are applicable. Indeed, since the channel matrix is symmetric and positive definite, the inverse channel matrix exists and also is symmetric. From Definition \ref{def:K}, all values of $K_j$ is the same since they are the same sum of permutation entries. Therefore, from Proposition \ref{prop:upperbound}, the optimal output probability mass vector  
\begin{equation}
q^*_j = \frac{2^{-K_j}}{\sum_i^n{2^{-K_i}}}
\end{equation} 
are equal each other for all $j$.  As a result, the input probability mass function $p^* = A^{-T}q^*$
is the uniform distribution, and the channel capacity is upper bounded by:
\begin{eqnarray}
C &\le& -\sum_{j=1}^n{q^*_j\log{q^*_j}} + \sum_{i=1}^{n}\sum_{j=1}^{n} {p_i ^*A_{ij}} \log {A_{ij}} \\
&=&\log n - H(A_{row}).
\end{eqnarray}

Interestingly, our result also shows the capacities of many channels that are {\em not} weakly symmetric, but admits the closed form formula of weakly symmetric channels.  In particular, consider a channel matrix called semi-weakly symmetric whose all rows are permutations of each other, but the sum of entries in each column might not be the same.  Furthermore, if the optimal condition is satisfied (Proposition \ref{prop:condition}), then the channel has closed-form capacity which is identical to the capacity of a symmetric and weakly symmetric channel:
\begin{equation}
\label{eq: weakly closed form}
C = \log n - H(A_{row}).
\end{equation}

For example, the following channel matrix:
	\[
	A=
	\begin{bmatrix}
	0.93   &   0.04  &   0.03 \\
	0.04   &  0.93  &   0.03\\
	0.04  &   0.03  &    0.93\\
	\end{bmatrix}
	\]
is not a weakly symmetric channel even though its rows are permutations of each other since the column sums are different.  However, this channel matrix satisfies Proposition \ref{prop:condition} and Corollary \ref{corollary 2} since $n=3$, ${\sigma_{\min}(A)}=0.88916$, $\sigma^*=0.825$, $H_{\max}(A)=0.43489$, $H_{\max}^*(A)=0.43592$ and $c_{\min}(A)=13.286$.
Thus, it has closed form formula for capacity, and can be easily shown to be $C = \log 3 - H(0.93, 0.04, 0.03)=1.1501$. The optimal output and input probability mass vectors can be shown to be:
\[
q^T=
  \begin{bmatrix}
    0.33333   &   0.33333   &   0.33333 
  \end{bmatrix},
\]
\[
p^T=
  \begin{bmatrix}
   0.32959  & 0.33337  & 0.33704
  \end{bmatrix},
\]
respectively. 


The following channel matrix is another example of semi-weakly symmetric matrix whose entries are controlled by a parameter $\gamma$ in the range of $(0,1)$ and given by the following form:
\begin{small}
$$\begin{array}{cc}
 \begin{bmatrix}
(1-\gamma)^3  & 3 (1-\gamma)^2 \gamma & 3(1-\gamma)\gamma^2 & \gamma^3\\
3 (1-\gamma)^2 \gamma & (1-\gamma)^3 &  \gamma^3 & 3(1-\gamma)\gamma^2\\
\gamma^3  & 3(1-\gamma)\gamma^2 & (1-\gamma)^3 & 3 (1-\gamma)^2 \gamma \\
\gamma^3 & 3(1-\gamma)\gamma^2 & 3 (1-\gamma)^2 \gamma & (1-\gamma)^3
\end{bmatrix}.
\end{array} $$
\end{small}

Fig. \ref{fig: weakly symmetric} shows the capacity upper bound of the semi-weakly symmetric channel and the actual channel capacity as function of $\gamma$.  As seen, for most of   $\gamma$, the upper bound is identical to the actual channel capacity which is numerically determined using CVX \cite{grant2008cvx}.

\begin{figure}
  \centering
  \includegraphics[width=3 in, height = 3 in]{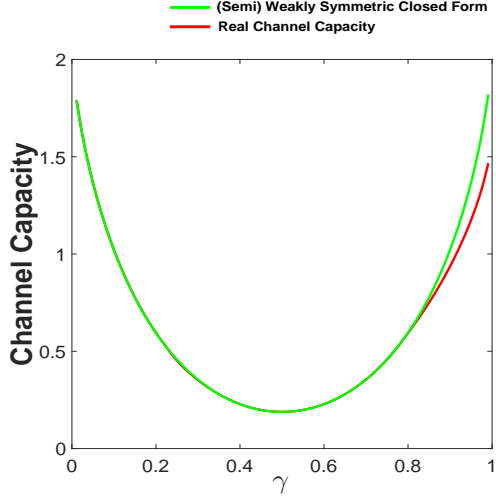}\\
  \caption{Channel capacity of  (semi) weakly symmetric channel as a function of  $\gamma$}\label{fig: weakly symmetric}
 \end{figure}

\subsection{Example 4: Unreliable Channels}
We now consider an unreliable channel whose channel matrix is:

\[
A=
  \begin{bmatrix}
0.6    &     0.3   &      0.1\\
         0.7     &     0.1     &    0.2\\        
         0.5     &    0.05    &     0.45 
  \end{bmatrix}.
\]

In this case, our optimality conditions do not satisfy, and the Arimoto upper bound is tightest ($0.17083$) as compared to our upper bound (0.19282) and Boyd-Chiang upper bound (0.848).

\subsection{Example 5: Bounds as Function of Channel Reliability}
Since we know that our proposed bounds are tight if the channel is reliable, we want to examine quantitatively how channel reliability affects various bounds.  In this example,  we consider a special class of channel whose channel matrix entries are controlled by a reliability parameter $\beta$ for  $0\leq \beta \leq 1$ as shown below: 
\[
A=
  \begin{bmatrix}
1-\beta  & 0.3 \beta & 0.4 \beta & 0.3 \beta\\
  0.4 \beta & 1-\beta & 0.3 \beta & 0.3 \beta\\
0.5\beta  & 0.4\beta & 1-\beta & 0.1 \beta\\
0.1\beta  & 0.2\beta & 0.7\beta & 1-\beta
  \end{bmatrix}.
\]

When $\beta$ is small, the channel tends to be reliable and when $\beta$ is large, the channel tends to be unreliable. 
Fig. \ref{fig: compare 2} shows various upper bounds as a function of $\beta$ together with the actual capacity.   The actual channel capacities for various $\beta$  are numerically computed using a convex optimization algorithm \cite{grant2008cvx}.   As seen, our closed form upper bound expression for capacity (red curve) from Proposition \ref{prop:upperbound} is much closer to the actual capacity (black dash curve) than other bounds for most values of $\beta$.  When $\beta$ is small ($\beta \leq 0.6$) or channel is reliable, the closed form upper bound is precise the real channel capacity, and we can verify that the optimal conditions in Proposition \ref{prop:condition} holds.  When the channel becomes unreliable, i.e.,  $\beta \geq 0.6$, our upper bound is no longer tight, however,  it is still the tightest among all the existing upper bounds. We note that when the $\beta$ is small, the channel matrix becomes a nearly diagonally dominant matrix, and our upper bound is tightest. 
\begin{figure}
  \centering
  \includegraphics[width=3 in, height =3 in]{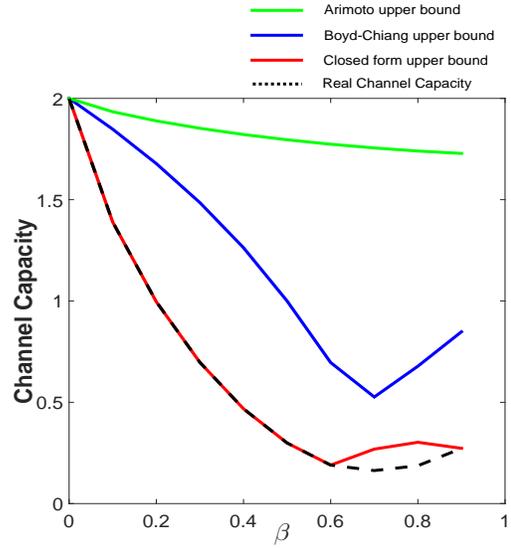}\\  
  \caption{Channel capacity and various upper bounds functions of  $\beta$}\label{fig: compare 2}
 \end{figure}
\section{Conclusion}
In this  paper, we describe an elementary technique based on Karush-Kuhn-Tucker (KKT) conditions to obtain (1) a good upper bound of a discrete memoryless channel having an invertible positive channel matrix and (2) a closed form expression for the capacity if the channel matrix satisfies certain conditions related to its singular value and its Gershgorin's disk. We provide a number of channels  where the proposed upper bound becomes precisely the capacity.  We also demonstrate that our proposed bounds are tighter than other existing bounds for these channels.

\medskip

\bibliographystyle{unsrt}
\bibliography{sample}

\begin{thebibliography}{10}

\bibitem{cover2012elements}
Thomas~M Cover and Joy~A Thomas.
\newblock {\em Elements of information theory}.
\newblock John Wiley \& Sons, 2012.

\bibitem{blahut1972computation}
Richard Blahut.
\newblock Computation of channel capacity and rate-distortion functions.
\newblock {\em IEEE transactions on Information Theory}, 18(4):460--473, 1972.

\bibitem{arimoto1972algorithm}
Suguru Arimoto.
\newblock An algorithm for computing the capacity of arbitrary discrete
  memoryless channels.
\newblock {\em IEEE Transactions on Information Theory}, 18(1):14--20, 1972.

\bibitem{muroga1953capacity}
Saburo Muroga.
\newblock On the capacity of a discrete channel, mathematical expression of
  capacity of a channel which is disturbed by noise in its every one symbol and
  expressible in one state diagram.
\newblock {\em Journal of the Physical Society of Japan}, 8(4):484--494, 1953.

\bibitem{shannon1956zero}
Claude Shannon.
\newblock The zero error capacity of a noisy channel.
\newblock {\em IRE Transactions on Information Theory}, 2(3):8--19, 1956.

\bibitem{robert1990ash}
B~Robert.
\newblock Ash. information theory, 1990.

\bibitem{nguyen2018closed}
Thuan Nguyen and Thinh Nguyen.
\newblock On closed form capacities of discrete memoryless channels.
\newblock In {\em 2018 IEEE 87th Vehicular Technology Conference (VTC Spring)},
  pages 1--5. IEEE, 2018.

\bibitem{martin2010algebraic}
Keye Martin, Ira~S Moskowitz, and Gerard Allwein.
\newblock Algebraic information theory for binary channels.
\newblock {\em Theoretical Computer Science}, 411(19):1918--1927, 2010.

\bibitem{liang2008algebraic}
Xue-Bin Liang.
\newblock An algebraic, analytic, and algorithmic investigation on the capacity
  and capacity-achieving input probability distributions of
  finite-input--finite-output discrete memoryless channels.
\newblock {\em IEEE Transactions on Information Theory}, 54(3):1003--1023,
  2008.

\bibitem{cotae2010eigenvalue}
Paul Cotae, Ira~S Moskowitz, and Myong~H Kang.
\newblock Eigenvalue characterization of the capacity of discrete memoryless
  channels with invertible channel matrices.
\newblock In {\em Information Sciences and Systems (CISS), 2010 44th Annual
  Conference on}, pages 1--6. IEEE, 2010.

\bibitem{grant2008cvx}
Michael Grant, Stephen Boyd, and Yinyu Ye.
\newblock Cvx: Matlab software for disciplined convex programming, 2008.

\bibitem{sinha2014convex}
Abhishek Sinha.
\newblock Convex optimization methods for computing channel capacity.
\newblock 2014.

\bibitem{dupuis2004blahut}
Fr{\'e}d{\'e}ric Dupuis, Wei Yu, and Frans~MJ Willems.
\newblock Blahut-arimoto algorithms for computing channel capacity and
  rate-distortion with side information.
\newblock In {\em Information Theory, 2004. ISIT 2004. Proceedings.
  International Symposium on}, page 179. IEEE, 2004.

\bibitem{matz2004information}
Gerald Matz and Pierre Duhamel.
\newblock Information geometric formulation and interpretation of accelerated
  blahut-arimoto-type algorithms.
\newblock In {\em Information theory workshop, 2004. IEEE}, pages 66--70. IEEE,
  2004.

\bibitem{yu2010squeezing}
Yaming Yu.
\newblock Squeezing the arimoto--blahut algorithm for faster convergence.
\newblock {\em IEEE Transactions on Information Theory}, 56(7):3149--3157,
  2010.

\bibitem{meister1967capacity}
Bernd Meister and Werner Oettli.
\newblock On the capacity of a discrete, constant channel.
\newblock {\em Information and Control}, 11(3):341--351, 1967.

\bibitem{jimbo1979iteration}
Masakazu Jimbo and Kiyonori Kunisawa.
\newblock An iteration method for calculating the relative capacity.
\newblock {\em Information and Control}, 43(2):216--223, 1979.

\bibitem{shannon1998mathematical}
Claude~E Shannon and Warren Weaver.
\newblock {\em The mathematical theory of communication}.
\newblock University of Illinois press, 1998.

\bibitem{cover1975achievable}
T~Cover.
\newblock An achievable rate region for the broadcast channel.
\newblock {\em IEEE Transactions on Information Theory}, 21(4):399--404, 1975.

\bibitem{boyd2004convex}
Stephen Boyd and Lieven Vandenberghe.
\newblock {\em Convex optimization}.
\newblock Cambridge university press, 2004.

\bibitem{weisstein2003gershgorin}
Eric~W Weisstein.
\newblock Gershgorin circle theorem.
\newblock 2003.

\bibitem{fiedler1967diagonally}
Miroslav Fiedler and Vlastimil Pt{\'a}k.
\newblock Diagonally dominant matrices.
\newblock {\em Czechoslovak Mathematical Journal}, 17(3):420--433, 1967.

\bibitem{thuan_thinh_2018}
Thuan Nguyen and Thinh Nguyen.
\newblock Relay-miso channel.
\newblock Available at \\ \url
  {http://ir.library.oregonstate.edu/concern/articles/tb09jb69h}, 2018.

\bibitem{cover1979capacity}
Thomas Cover and A~EL Gamal.
\newblock Capacity theorems for the relay channel.
\newblock {\em IEEE Transactions on information theory}, 25(5):572--584, 1979.

\bibitem{rankov2006achievable}
Boris Rankov and Armin Wittneben.
\newblock Achievable rate regions for the two-way relay channel.
\newblock In {\em Information theory, 2006 IEEE international symposium on},
  pages 1668--1672. IEEE, 2006.

\bibitem{chiang2004geometric}
Mung Chiang and Stephen Boyd.
\newblock Geometric programming duals of channel capacity and rate distortion.
\newblock {\em IEEE Transactions on Information Theory}, 50(2):245--258, 2004.

\bibitem{boyd2007tutorial}
Stephen Boyd, Seung-Jean Kim, Lieven Vandenberghe, and Arash Hassibi.
\newblock A tutorial on geometric programming.
\newblock {\em Optimization and engineering}, 8(1):67, 2007.

\bibitem{petersen2008matrix}
Kaare~Brandt Petersen, Michael~Syskind Pedersen, et~al.
\newblock The matrix cookbook.
\newblock {\em Technical University of Denmark}, 7(15):510, 2008.

\bibitem{Rayleigh_quotient}
Rayleigh quotient and the min-max theorem.
\newblock Available at \\ \url
  {http://www.math.toronto.edu/mnica/hermitian2014.pdf}, 2014.

\bibitem{johnson1989gersgorin}
Charles~R Johnson.
\newblock A gersgorin-type lower bound for the smallest singular value.
\newblock {\em Linear Algebra and its Applications}, 112:1--7, 1989.

\bibitem{hong1992lower}
YP~Hong and C-T Pan.
\newblock A lower bound for the smallest singular value.
\newblock {\em Linear Algebra and its Applications}, 172:27--32, 1992.

\end{thebibliography}
\appendix

\subsection{Proof of Lemma \ref{lemma:dominant}}
\label{sec: appendix for lemma 1}
For claim (a), since the channel matrix is strictly diagonally dominant, using Gershgorin circle theorem \cite{weisstein2003gershgorin} that for any eigenvalues $\lambda_1, \lambda_2, \dots, \lambda_n$, we must have:
\begin{equation*}
\lambda_i \geq  A_{ii} - \sum_{j \neq i}|A_{ij}| > 0.
\end{equation*}

Thus, $det(A) = \lambda_1\lambda_2\dots\lambda_n > 0$. Therefore, $A$ is invertible.

Claim (b) is a well-known algebra result \cite{petersen2008matrix}. 
%
%

For claim (c), due to $AA^{-1}=I$ and $A_{ij}>0$ $\forall$ $i,j$, therefore, for $\forall$ $j$ exists at least $i$ such that $A^{-1}_{ij} \neq 0$.  Therefore the largest absolute entry in each column $\neq 0$.
Claim (c) can be obtained by contradiction. Suppose that the largest absolute entry in $j^{th}$ column of $A^{-1}$ is   $A^{-1}_{ij}$ in $i^{th}$ row, that said $|A^{-1}_{ij}| \geq |A^{-1}_{kj}|$ for $\forall$ $k$. We suppose that 
 $A^{-1}_{ij} <0$. Thus:
\begin{small}
\begin{eqnarray}
\sum_{k=1}^{n}A_{ik}A^{-1}_{kj} &\leq& -A_{ii}|A^{-1}_{ij}|+\sum_{k=1,k\neq i}^{n} A_{ik}|A^{-1}_{ij}| \label{eq: app3}\\
&=& (-A_{ii}+\sum_{k=1,k\neq i}^{n} A_{ik})|A^{-1}_{ij}| \nonumber\\
&<&0 \label{eq: app4},
\end{eqnarray}
\end{small}
which contradicts with $\sum_{k=1}^{n}A_{ik}A^{-1}_{kj}=I_{ij} \geq 0$. Thus, the largest absolute value in each column of $A^{-1}$ is positive. That said in $j^{th}$ column, if $|A^{-1}_{ij}| \geq |A^{-1}_{kj}|$ for $\forall$ $k$, then $A^{-1}_{ij} > 0$.

Now, suppose that the largest absolute element in $j^{th}$ column of $A^{-1}$, is $A^{-1}_{ij}$ with $i \neq j$ and $A^{-1}_{ij} >0$. Then:
\begin{small}
\begin{eqnarray}
0&=&\sum_{k=1}^{n}A_{ik}A^{-1}_{kj} \nonumber\\
&\geq& A_{ii}|A^{-1}_{ij}|-\sum_{k=1,k\neq i}^{n} A_{ik}|A^{-1}_{ij}| \label{eq: app1}\\
&=& (A_{ii}-\sum_{k=1,k\neq i}^{n} A_{ik})A^{-1}_{ij} \nonumber\\
&>&0 \label{eq: app2},
\end{eqnarray}
\end{small}
with  (\ref{eq: app1}) due to $A^{-1}_{ij}$ is the largest absolute element in $j^{th}$ column and  (\ref{eq: app2}) due to $A$ is strictly diagonally dominant matrix. This is a contradiction.  Therefore, the largest absolute entry in $j^{th}$ column of $A^{-1}$ should be $A^{-1}_{jj}$ and $A^{-1}_{jj} >0$. 

\subsection{Proof of Lemma \ref{lemma:c_i}}
\label{sec: appendix for lemma 2}
First, let's show that the second largest absolute value in each column of $A^{-1}$ is a negative entry by contradiction method. Suppose that the second largest absolute value in $j^{th}$ column of $A^{-1}$ is positive and in $k^{th}$ row ($k \neq j$), $A_{kj}^{-1} \geq 0$. Consider,
\begin{small}
\begin{eqnarray}
0 &=&\sum_{i=1}^{n}A_{ki}A_{ij}^{-1} \nonumber\\
 &\geq& A_{kj}A_{jj}^{-1} + A_{kk}A_{kj}^{-1} - |\sum_{i=1, i \neq k;i \neq j}^{n}A_{ki}A_{ij}^{-1}| \label{eq: 328}\\
  &\geq& A_{kj}A_{jj}^{-1} + A_{kk}A_{kj}^{-1} - \sum_{i=1, i \neq k;i \neq j}^{n}|A_{ki}A_{ij}^{-1}| \label{eq: 329}\\
 &\geq& A_{kj}A_{jj}^{-1} + A_{kk}A_{kj}^{-1} - \sum_{i=1, i \neq k;i \neq j}^{n}A_{ki}|A_{ij}^{-1}| \label{eq: 331}\\
 &\geq& A_{kj}A_{jj}^{-1} + A_{kk}A_{kj}^{-1} - \sum_{i=1, i \neq k;i \neq j}^{n}A_{ki}|A_{kj}^{-1}| \label{eq: 330}\\
 &=& A_{kj}A_{jj}^{-1} + A_{kj}^{-1}( A_{kk}-\sum_{i=1, i \neq k;i \neq j}^{n}A_{ki}) \label{eq: 332} \\
&>& 0 \label{eq: 340},
\end{eqnarray} 
\end{small}
with  (\ref{eq: 328}) due to the fact that $C \geq -|C|$ for $\forall$ $C$,  (\ref{eq: 329}) due to the triangle inequality,  (\ref{eq: 331}) due to $A_{ki}$ is positive,  (\ref{eq: 330}) due to $A_{kj}^{-1}$ is the second largest absolute value in $j^{th}$ column of $A^{-1}$,  (\ref{eq: 332}) due to the assumption that $A_{kj}^{-1} \geq 0$ and   (\ref{eq: 340}) due to  (\ref{eq: condition for diagonal}) such that $A_{kk} \geq \sum_{i=1, i \neq k}^{n}A_{ki} \geq \sum_{i=1, i \neq k;i \neq j}^{n}A_{ki}$. Thus, the second largest absolute value in column of $A^{-1}$ is negative ($A_{kj}^{-1} <0$). Due to Lemma \ref{lemma:dominant} part $c$, $A_{jj}^{-1}$ is the largest absolute value entry and $A_{jj}^{-1} >0$. Similarly,
\begin{small}
\begin{eqnarray}
0 &=&\sum_{i=1}^{n}A_{ki}A_{ij}^{-1} \nonumber\\
 &\leq& A_{kj}A_{jj}^{-1} + A_{kk}A_{kj}^{-1} + |\sum_{i=1, i \neq k;i \neq j}^{n}A_{ki}A_{ij}^{-1}| \label{eq: 343}\\
 &\leq& A_{kj}A_{jj}^{-1} + A_{kk}A_{kj}^{-1} + \sum_{i=1, i \neq k;i \neq j}^{n}|A_{ki}A_{ij}^{-1}| \label{eq: 344}\\
 &\leq& A_{kj}A_{jj}^{-1} + A_{kk}A_{kj}^{-1} + \sum_{i=1, i \neq k;i \neq j}^{n}A_{ki}|A_{ij}^{-1}| \label{eq: 341}\\
 &\leq& A_{kj}A_{jj}^{-1} - A_{kk}|A_{kj}^{-1}| + \sum_{i=1, i \neq k;i \neq j}^{n}A_{ki}|A_{kj}^{-1}| \label{eq: 342}
\end{eqnarray} 
\end{small}
with  (\ref{eq: 343}) due to  the fact that $C \leq |C|$ for $\forall$ $C$,  (\ref{eq: 344}) due to the triangle inequality,  (\ref{eq: 341}) due to $A_{ki} \geq 0$, $\forall$ $i$ and   (\ref{eq: 342}) due to $A_{kj}^{-1} <0$ and $A_{kj}^{-1}$ is the second largest absolute value in $j^{th}$ column. Hence,
\begin{small}
\begin{eqnarray}
 A_{kj}A_{jj}^{-1} &\geq& A_{kk}|A_{kj}^{-1}| - \sum_{i=1, i \neq k;i \neq j}^{n}A_{ki} |A_{kj}^{-1}|  \nonumber\\
 A_{jj}^{-1} &\geq& \dfrac{|A_{kj}^{-1}|  (A_{kk}  - \sum_{i=1, i \neq k;i \neq j}^{n}A_{ki})}{A_{kj}} \nonumber\\
  A_{jj}^{-1} &\geq& |A_{kj}^{-1}| \dfrac{A_{kk}-\dfrac{A_{kk}}{c_{\min}(A)}} {\dfrac{A_{kk}}{c_{\min}(A)}} \label{eq: 36}\\
   A_{jj}^{-1} &\geq&   |A_{kj}^{-1}| [c_{\min}(A)-1] \label{eq: 36-b},
\end{eqnarray}
\end{small}
for $\forall$ $j$, with  (\ref{eq: 36}) due to Definition \ref{def:gershgorin} and  (\ref{eq: c min A}) such that $\dfrac{A_{kk}}{c_{\min}(A)} \geq \sum_{i=1, i \neq k}^{n}A_{ki} \geq \sum_{i=1, i \neq k, i \neq j}^{n}A_{ki}$. Thus, we have:
\begin{small}
\begin{eqnarray}
\label{eq: c_i AT}
c_j(A^{-T}) = \dfrac{A_{jj}^{-1}}{\sum_{ k \neq j}^{}|A_{kj}^{-1}|} \geq \dfrac{c_{\min}(A)-1}{n-1}.
\end{eqnarray}
\end{small}

Thus,  (\ref{eq: thuan1419}) is proven. 

Next, we note that from  (\ref{eq: 36-b})
\begin{small}
\begin{equation}
 \dfrac{A_{jj}^{-1}}{c_{\min}(A)-1} \geq   |A_{kj}^{-1}|, 
\end{equation}
\end{small}
for $\forall$ $k$. Moreover, from Lemma \ref{lemma:dominant}, $A_{jj}^{-1} \geq 0$ and is the largest entry in $j^{th}$ row. Thus, for an arbitrary $L$ and $S$,
\begin{small}
\begin{eqnarray}
\label{eq: thuan1416}
|A^{-1}_{Lj}|+|A^{-1}_{Sj}| &\leq & A_{jj}^{-1} + \dfrac{A_{jj}^{-1}}{c_{\min}(A)-1} \nonumber\\
&=&  A_{jj}^{-1}\dfrac{c_{\min}(A)}{c_{\min}(A)-1},
\end{eqnarray}
\end{small}
for $\forall$ $j$. 
Thus,  (\ref{eq: thuan1418}) is proven. 

\subsection{Proof of Lemma \ref{lemma: singular}}
\label{sec: appendix for lemma 3}

Consider the matrix $B=A^{-1}A^{-T}$, $B$ is symmetric, all its eigenvalues  are real and satisfy the Rayleigh quotient \cite{Rayleigh_quotient}. Let $\lambda_B^{max}$ be the maximum eigenvalue of $B$ then from \cite{Rayleigh_quotient}
\begin{equation}
\label{eq: Raleigh}
R(B,x)=\dfrac{x^*Bx}{x^*x} \leq \lambda^{max}_B.
\end{equation}

Consider the unit vector $e=[0,\dots,1,\dots,0]^T$ with entry ``1" is in the $i^{th}$ column. Let $x = e$ in  (\ref{eq: Raleigh}), we have:
\begin{equation}
B_{ii} \leq \lambda_B^{max}.
\end{equation}

Thus,
\begin{eqnarray}
\lambda_B^{max} &\geq& B_{ii} \nonumber\\
 &=&\sum_{j=1}^{n}A^{-1}_{ij}A^{-1}_{ij} \nonumber\\
&\geq& {(A^{-1}_{ii})}^{2}.
\end{eqnarray}

Now since $B$ is a symmetric matrix $\lambda_B^{max}={\sigma_{\max}(B)}$ \cite{petersen2008matrix}. However, from \cite{petersen2008matrix}, ${\sigma_{\max}(B)}={\sigma_{\max}(A^{-1}}A^{-T})=\sigma_{\max}^2A^{-1}$ and ${\sigma_{\max}A^{-1}}=\dfrac{1}{{\sigma_{\min}(A)}}$. Thus:
\begin{equation}
\dfrac{1}{{\sigma_{\min}(A)}} \geq A^{-1}_{ii}.
\end{equation}

From Lemma \ref{lemma:dominant}-$c$, the largest entry in $A^{-1}$ must be a diagonal element, thus 
$$\max_{i,j}{A^{-1}_{ij}} \leq \dfrac{1}{{\sigma_{\min}(A)}}.$$

\subsection{Proof of Lemma \ref{prop:sum_rowAinv=1}}
\label{appendix for pmf}

For the first claim, since $A$ is a stochastic matrix,
$$A\textbf{1} = \textbf{1}.$$
Left multiply both sides by $A^{-1}$ results in
$\textbf{1} = A^{-1}\textbf{1}.$
For the second claim, left multiplying $y = {A^{-T}}x$ by $\textbf{1}^T$, we have:
$$\textbf{1}^Ty = \textbf{1}^T{A^{-T}}x = x^TA^{-1}\textbf{1} = x^T\textbf{1} = 1,$$
where we use $A^{-1}\textbf{1} = \textbf{1}$ in the previous claim.

Thus, we have $\sum_{i}^{n}p_i^*=1$ since from   (\ref{eq:q}), $q^*$ is a probability mass vector.  

\subsection{Proof of Corollary \ref{corollary 2}}
\label{appendix for upper bound singular value}
\begin{lemma}
\label{lemma: 4}
Lower bound of $\sigma_{\min}(A)$ and upper bound of $H_{\max}(A)$ are $\sigma^*$ and $H_{\max}^*(A)$, respectively
\begin{equation}
\label{eq: vy1000}
\sigma_{\min}(A) \geq \sigma^*=\dfrac{c_{\min}(A) -n/2}{c_{\min}(A)+1},
\end{equation}
and 
\begin{eqnarray}
\label{eq: vy1001}
H_{\max}(A) \leq H_{\max}^*(A),
\end{eqnarray}
where
\begin{equation}
H_{\max}^*(A)\!=\!\log(c_{\min}(A)+1)\!+\!\dfrac{\log(n\!-\!1)\!-\!c_{\min}(A)\log c_{\min}(A)}{c_{\min}(A)+1}.
\end{equation}

\end{lemma}
\begin{proof}

Due to the channel matrix is a strictly diagonally dominant positive matrix. Thus, we have
\begin{equation}
\label{eq: 201}
A_{kk} \geq \dfrac{c_{\min}(A)}{c_{\min}(A) +1},
\end{equation}
\begin{equation}
\label{eq: 202}
R_k(A)=1-A_{kk} \leq 1-\dfrac{c_{\min}(A)}{c_{\min}(A) +1}=\dfrac{1}{c_{\min}(A) +1},
\end{equation}
\begin{equation}
\label{eq: 203}
C_k(A)=\sum_{j=1,j\neq k}^{j=n}A_{jk} \leq \sum_{j=1,j\neq k}^{j=n} R_j(A) \leq \dfrac{n-1}{c_{\min}(A) +1},
\end{equation}

for $\forall$ $k$ with  (\ref{eq: 201}) due to  (\ref{eq: find c min}),  (\ref{eq: 202}) due to  (\ref{eq: 201}),  (\ref{eq: 203}) due to the fact that $\forall$ $j \neq k$, $A_{jk} \leq \sum_{j\neq k}A_{jk}=R_j(A)$  and each $R_{j}(A) \leq \dfrac{1}{c_{\min}(A) +1}$ which is proven in  (\ref{eq: 201}). 
Now, we are ready to establish the upper bound of $H_{\max}(A)$ and the lower bound of $\sigma_{\min}(A)$, respectively. 

$
\bullet$
Suppose that $H_{\max}(A)$ achieves at  $k^{th}$ row, then
\begin{small}
\begin{eqnarray}
\label{eq: upper bound H max}
H_{\max}(A) &\!=\!&\!-\!(\sum_{i=1}^{n}A_{ki}\log A_{ki}) \nonumber\\
&\!=\!&- (A_{kk}\log A_{kk} +\sum_{i=1,i \neq k}^{n}A_{ki}\log A_{ki} ) \nonumber\\
&\!=\!&- A_{kk}\log A_{kk}  \nonumber\\
&\!-\!& (1\!-\!A_{kk})\sum_{i\!=\!1,i \!\neq \!k}^{n}\dfrac{A_{ki}}{1\!-\!A_{kk}}(\log \dfrac{A_{ki}}{1\!-\!A_{kk}}\!+\!\log (1\!-\!A_{kk})) \nonumber\\
&\!=\!&- A_{kk}\log A_{kk} \nonumber\\
&-& (1-A_{kk})\sum_{i=1,i \neq k}^{n}\dfrac{A_{ki}}{1-A_{kk}}\log \dfrac{A_{ki}}{1-A_{kk}} \nonumber\\
&-& (1-A_{kk})\log (1-A_{kk}) \nonumber\\
&\!\leq\!&- A_{kk}\log A_{kk} + (1-A_{kk})\log(n-1) \nonumber\\
&-& (1-A_{kk})\log (1-A_{kk}) \label{eq: vy1003}\\
&\!=\!&- (A_{kk}\log A_{kk} + (1-A_{kk})\log(\dfrac{1-A_{kk}}{n-1})) \nonumber\\
&\! \leq \!&- (\dfrac{c_{\min}(A)}{c_{\min}(A) +1}\log\dfrac{c_{\min}(A)}{c_{\min}(A) +1} \nonumber\\
&\!+\!& (1-\dfrac{c_{\min}(A)}{c_{\min}(A) +1})\log\dfrac{1-\dfrac{c_{\min}(A)}{c_{\min}(A) +1}}{n-1}) \label{eq: vy1004}\\
&\!=\!& \log(c_{\min}(\!A\!)\!+\!1)\!+\!\dfrac{\log(n\!-\!1)\!-\!c_{\min}(\!A\!)\log c_{\min}(\!A\!)}{c_{\min}(A)+1} \nonumber,
\end{eqnarray}
\end{small}
with (\ref{eq: vy1003}) is due to $-\sum_{i=1,i \neq k}^{n}\dfrac{A_{ki}}{1-A_{kk}}\log \dfrac{A_{ki}}{1-A_{kk}}$
 is the entropy of $n-1$ elements which is bounded by $\log(n-1)$. For (\ref{eq: vy1004}), first we show that 
$f(x)=- (x\log x + (1-x)\log(\dfrac{1-x}{n-1}))$ is monotonically decreasing function for $\dfrac{x}{1-x} \geq n-1$. Indeed,
\begin{eqnarray}
\dfrac{d(f(x))}{d(x)}&=& \log x-\log(1-x)-\log(n-1) \nonumber\\
&=&-(\log\dfrac{x}{1-x} - \log (n-1)) \nonumber.
\end{eqnarray}

Thus, if $\dfrac{x}{1-x} \geq n-1$ then $\dfrac{d(f(x))}{d(x)} \leq 0$. However,
from (\ref{eq: 201}), 
\begin{equation}
\label{eq: vy1010}
\dfrac{A_{kk}}{1-A_{kk}} \geq \dfrac{\dfrac{c_{\min}(A)}{c_{\min}(A) +1}}{1-\dfrac{c_{\min}(A)}{c_{\min}(A) +1}} =c_{\min}(A).
\end{equation}

From (\ref{eq:eigen2})
\begin{equation}
\label{eq: vy1009}
c_{\min}(A) \geq 1+(n-1)^2 2^{\frac{n  H_{\max}^*(A)}{{\sigma^*}}} \geq 1+(n-1)^2 > n-1,
\end{equation}
due to $\frac{n  H_{\max}^*(A)}{{\sigma^*}} \geq 0$ and  $n \geq 2$. Thus, $\dfrac{A_{kk}}{1-A_{kk}} >n-1$. 
From (\ref{eq: vy1010}) and (\ref{eq: vy1009}), $f(x)$ is decreasing function and  (\ref{eq: vy1004}) is constructed by plugging the lower bound of $A_{kk}$ in (\ref{eq: 201}). 

$
\bullet$
Secondly, the lower bound of $\sigma_{\min}(A)$ can be found in \cite{johnson1989gersgorin} (Theorem 3)
\begin{equation}
\label{eq: upper bound singular 1}
\sigma_{\min}(A) \geq \min_{1\leq k \leq n} |A_{kk}|-\dfrac{1}{2} (R_{k}(A) + C_{k}(A)),
\end{equation}

or in \cite{hong1992lower} (Theorem 0)
\begin{small}
\begin{equation}
\label{eq: upper bound singular 2}
\!\sigma_{\min}(\!A\!) \!\geq \!\min_{1\leq k \leq n}  \dfrac{1}{2}( \{4|A_{kk}|^2 \!+\!(R_k(\!A\!)\!-\!C_k(\!A\!))^2  \}^{1/2} \!-\![ R_k(\!A\!)\!+\!C_k(\!A\!) ] ),
\end{equation}
\end{small}
with $R_k(A)=\sum_{j=1,j\neq k}^{j=n}|A_{kj}|$ and $C_k(A)=\sum_{j=1,j\neq k}^{j=n}|A_{jk}|$, respectively. 
Thus, if we use the lower bound established in  (\ref{eq: upper bound singular 2}), 
\begin{eqnarray}
\sigma_{\min}(A) & \geq& \dfrac{1}{2}( \{4  [\dfrac{c_{\min}(A)}{c_{\min}(A) +1}]^2  \}^{1/2} \nonumber\\
&-& [ \dfrac{1}{c_{\min}(A) +1} +\dfrac{n-1}{c_{\min}(A) +1} ] ) \label{eq: 204}\\
&=& \dfrac{c_{\min}(A) -n/2}{c_{\min}(A)+1}=\sigma^* \nonumber,
\end{eqnarray}
with  (\ref{eq: 204}) due to  (\ref{eq: 201}),  (\ref{eq: 202}),  (\ref{eq: 203}) and the fact that $\{R_k(A)-C_k(A)\}^2 \geq 0$. 

A similar lower bound can be constructed using  (\ref{eq: upper bound singular 1})
\begin{eqnarray}
\sigma_{\min}(A) & \geq& \dfrac{c_{\min}(A)}{c_{\min}(A) +1} \nonumber\\
&-& \dfrac{1}{2} (\dfrac{1}{c_{\min}(A) +1} +\dfrac{n-1}{c_{\min}(A) +1}) \label{eq: 205}\\
&=& \dfrac{c_{\min}(A) -n/2}{c_{\min}(A)+1}=\sigma^* \nonumber,
\end{eqnarray}
with  (\ref{eq: 205}) due to  (\ref{eq: 201}),  (\ref{eq: 202}) and  (\ref{eq: 203}).
As seen, both our approaches yield a same lower bound of $\sigma_{\min}(A)$.  However,  (\ref{eq: upper bound singular 2}) is  tighter than  (\ref{eq: upper bound singular 1}) due to $\{R_k(A)-C_k(A)\}^2$.

\end{proof}
\end{document}